\documentclass{article}

\usepackage{arxiv}

\usepackage[utf8]{inputenc} 
\usepackage[T1]{fontenc}    
\usepackage{hyperref}       
\usepackage{url}            
\usepackage{booktabs}       
\usepackage{amsfonts}       
\usepackage{nicefrac}       
\usepackage{microtype}      
\usepackage{lipsum}		
\usepackage{graphicx}
\usepackage{natbib}
\usepackage{doi}

\usepackage{graphicx}
\usepackage{mathptmx}
\usepackage[T1]{fontenc}

\usepackage{paralist}
\usepackage{comment}

\usepackage{latexsym}
\usepackage{graphicx}
\usepackage{mathptmx}
\usepackage[T1]{fontenc}
\usepackage{float}
\usepackage{comment}
\usepackage{xcolor}

\usepackage{amsmath}
\usepackage{amsfonts}
\usepackage{amssymb}
\usepackage{amsbsy}
\usepackage{amsthm}
\usepackage{array}
\usepackage[ruled, noend]{algorithm2e}
\def\E{{\mathbb E}}
\def\Var{{\rm Var}}

\usepackage{mathtools}
\def\argmax{\mathop{\rm arg\,max}}%
\def\argmin{\mathop{\rm arg\,min}}%
\def\diag{\text{diag}}
\allowdisplaybreaks
\newtheorem{lemma}{Lemma}

\title{\Large Digital Twin Calibration for Biological System-of-Systems: \\
Cell Culture Manufacturing Process }

\author{
\href{https://orcid.org/0009-0005-4356-4345}{\includegraphics[scale=0.06]{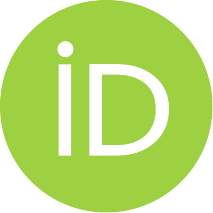}\hspace{1mm}Fuqiang Cheng} \\
	Northeastern University\\
\\
\And
\href{https://orcid.org/0000-0001-9563-4927}{\includegraphics[scale=0.06]{orcid.pdf}\hspace{1mm}Wei Xie}\thanks{Corresponding author. Email: \url{w.xie@northeastern.edu}} \\
	Northeastern University\\
	\And
\href{https://orcid.org/0000-0001-9555-7132}{\includegraphics[scale=0.06]{orcid.pdf}\hspace{1mm}Hua Zheng} \\
Northeastern University\\
\date{}
}

\renewcommand{\shorttitle}{\textbf{\small Cheng et al.:} {\small Digital Twin Calibration for Biological System-of-Systems:  Cell Culture Manufacturing Process }}

\hypersetup{
pdftitle={Structure-Function Dynamics Hybrid Modeling:
RNA Degradation},
pdfauthor={Zheng et al. 2023},
}

\begin{document}
\maketitle

\begin{abstract}
Biomanufacturing innovation relies on an efficient Design of Experiments (DoEs) to optimize processes and product quality. Traditional DoE methods, ignoring the underlying bioprocessing mechanisms, often suffer from a lack of interpretability and sample efficiency. This limitation motivates us to create a new optimal learning approach for digital twin model calibration. In this study, we consider the cell culture process multi-scale mechanistic model, also known as Biological System-of-Systems (Bio-SoS).
This model with a modular design, composed of sub-models, allows us to integrate data across various production processes.  
To calibrate the Bio-SoS digital twin, we evaluate the mean squared error of model prediction and develop a computational approach to quantify the impact of parameter estimation error of individual sub-models on the prediction accuracy of digital twin, which can guide sample-efficient and interpretable DoEs.
\end{abstract}


\section{INTRODUCTION}
\label{sec:intro}

To support interpretable predictions and optimal control of biomanfuacturing processes, in this paper, we develop a digital twin calibration approach for multi-scale bioprocess mechanistic model or Biological System-of-Systems (Bio-SoS) \citep{Zheng2024} characterizing causal interdependence from molecular- to cellular- to macro-kinetics. Even though this study is motivated by cell culture process, it can be extended to calibrate general Bio-SoS with modular design. Basically, cell culture process dynamics and variations depend on the modules: (1) a single cell mechanistic model characterizing each living cell behaviors and their interactions with environment; (2) a metabolic shift model characterizing the change of cell metabolic phase and behaviors 
as a response to culture conditions and cell age; and (3) macro-kinetic model of a bioreactor system composed of many living cells under different metabolic phases.    

The benefits of considering the Bio-SoS mechanistic model with \textit{modular design} include: a) support flexible manufacturing through assembling a system of modules to account for biomanufacturing processes under different conditions and inputs; and b) facilitate the integration of heterogeneous data from different production processes, such as 2D culture and 3D aggregate culture for Induced Pluripotent Stem Cells (iPSCs) \citep{Keqi2024, Zheng2024}. By incorporating the structure property of the Bio-SoS mechanistic model into the calibration method, we can quantify how the model uncertainties or approximation errors of different modules interact with each other and propagate through the reaction pathways to the prediction of outputs (e.g., yield and product quality attributes), which can guide  interpretable and most informative {Design of Experiments (DoEs)} to efficiently improve model fidelity with less experiments.


The model uncertainty quantification approaches for digital twin calibration can be divided into two main categories: Bayesian and frequentist approaches \citep{CORLU2020}. Bayesian approaches treat unknown model parameters as random variables and quantify our belief by posterior distributions. 
It involves specifying prior distributions for model parameters and updating these distributions based on the information from observed data by applying Bayes' theorem. On the other hand, frequentist approaches rely on traditional statistical methods, such as maximum likelihood estimation, to find parameter values that maximize the likelihood of observed data. 
Those estimation approaches could be faster under the situations when there not exist conjugate priors and posterior sampling update is computationally expensive.

For the propagation of input and model uncertainty to simulation outputs, to save computational budget and time, a metamodel is often used to approximate the response surface, especially for complex stochastic systems. In the classical calibration approach \citep{kennedy2001}, {Gaussian Processes (GPs) }are often used to model the mean response and discrepancies, deriving the posterior distribution for predictions and guiding the design of experiments. 
Although GP metamodel is often used to propagate from inputs to outputs (e.g., mean response),
it's hard to interpret and difficult to leverage prior knowledge of the real system mechanisms (such as bioprocessing reaction network structure). To enhance interpretability and sample efficiency, the mechanistic model structure will be employed to construct the response surface in our study.

Calibration criteria play a critical role to guide DoEs for digital twin calibration. 
Mean response and Mean Squared Error (MSE) are widely used criteria for assessing model prediction accuracy. \cite{Wu2014} shows Kennedy's method, which models the mean output as a Gaussian process, may lead to asymptotically $L_2$-inconsistent. They modify the criteria into $L_2$ norm of the discrepancy between two system outputs and propose a $L_2$-consistent calibration method, which has 
an optimal convergence rate.
  

To support process prediction and optimal control, in this paper, we aim to calibrate the Bio-SoS mechanistic model with a modular design to improve its predictive accuracy and reduce the MSE of the process output prediction. We employ the Maximum Likelihood Estimation (MLE) method to estimate model parameters and utilize bootstrap techniques to quantify estimation errors across different modules. To optimize the digital twin calibration policy and guide the most informative data collection, we derive a gradient-based approach that follows the steepest descent search in policy parameter space, making the learning process more interpretable. 
This strategy utilizes the Linear Noise Approximation (LNA) for uncertainty propagation and constructs a surrogate model for MSE. Concurrently, we preserve mechanistic information by solving the Ordinary Differential Equations (ODEs) from LNA using Euler's method. In an online setting, a calibration policy is iteratively updated by using a stochastic gradient method where the gradient of MSE with respect to the calibration policy parameter follows the backward direction of the uncertainty propagation. The proposed calibration approach accounts for the Bio-SoS mechanism structure and quantifies Bio-SoS module error interaction and their propagation through mechanistic pathways to output prediction, which guides the optimal sequential {DoEs} to efficiently improve model fidelity and prediction accuracy.

In sum, we develop an interpretable and sample efficient calibration approach for a multi-scale bio-process mechanistic model so that the digital twin of the cell culture process can improve process prediction and support optimal control. Even though this paper uses cell culture as a motivation example of Bio-SoS, the proposed calibration approach can be extendable to general biomanufacturing systems. The key contributions of the proposed calibration approach and the benefits are summarized as follows.
\begin{itemize}
    \item We proposed a new sequential DoE method for calibrating a multi-scale bioprocess mechanistic model by updating the parameters using closed-form gradients that incorporate the mechanistic information. 
    
    \item  We assess the MSE of the process output prediction and use a LNA-based metamodel along with Euler's method to estimate how model uncertainty propagates through Bio-SoS mechanism pathways and impacts on the output prediction accuracy. This approach can advance our understanding on how the errors in individual module parameters affect the overall accuracy of digital twin model prediction.
    
    \item Built on the LNA, we further develop a gradient-based policy optimization to guide most informative dsign of experiments and support sample-efficient optimal learning. 
\end{itemize} 

The organization of the paper is as follows. In Section~\ref{sec: MotivationExample}, 
we describe the multi-scale mechanistic model. Then, we propose the novel calibration method for the multi-scale model in Sections~\ref{sec:calibration}, which includes two main procedures: model inference and policy update. A stochastic gradient method is developed for the model inference in Section~\ref{sec:MLE}. 
The calibration policy gradient estimation and update method is further developed in Section~\ref{sec:optimal learning}. Then we empirically validate our method in Section~\ref{sec: empirical study} and conclude this paper in Section~\ref{sec: conclusion}.


\section{Problem Description and Bio-SoS Mechanistic Model}
\label{sec: MotivationExample}


In this study, we explore a multi-scale bioprocess mechanistic model designed to elucidate the interdependencies that exist across molecular, cellular, and macroscopic levels within a cell culture process. This comprehensive model integrates the complexities of biological interactions within and between these scales and extends applicably to a broader range of biological systems, referred to as biological system-of-systems (Bio-SoS). A detailed summary of the variable notations used within our model is provided in Table~\ref{table:variables}.

Our model is based on the structured Markov chain, capturing both the dynamics and variability inherent in cellular processes through a state transition probabilistic model. Each individual cell operates as a complex system, and collectively, numerous cells at various metabolic phases within the bioreactor form an intricate system of systems. This interaction is depicted in Figure~\ref{figure:Bio-SoS}, where three distinct metabolic phases are considered. Cells interact with each other by altering their environment, through nutrient uptake and the production of metabolic wastes, and in turn, respond to these environmental changes.

\begin{table}[h!]
\centering
\begin{tabular}{| p{0.5cm} | p{6.5cm} | p{0.5cm}  | p{7.3cm} |}
\hline
$\Tilde{\mathbf{s}}$ & Single cell state & $\pmb s$ & Macro-state \\
$r$ & Reaction indices ($r=1,2,\ldots,R$) & $Z_t$ & Cell phases for single cell ($Z_t=i, i=0,1,\ldots,I$) \\
$X_{i,t}$ & Cell density in $i$-th phase ($i=0,1,\ldots,I$) & $m$ & State component indices ($m=1,2,\ldots,M$) \\
$t$ & Time ($t=0,1,2,\ldots,T$) & $k$ & Calibration iteration number ($k=1,2\ldots,K$) \\
$\mathbb{P}_{ii'}$ & Transition probability from $i$- to $i'$-th phase & $\pmb\beta$ & Parameters for phase shift model \\
$\pmb\alpha_i$ & Model parameters for cells in phase $i$ & $\pmb\theta^c$ & True model parameters \\
$\widehat{\pmb\theta}_k$ & Estimated model parameters at iteration $k$ & $\pmb\omega_k$ & Policy parameters \\
$\lambda$ & Learning rate for updating parameter & $\gamma_k$ & Learning rate for updating policy  \\
$\pmb{u}_k$ & Virtual model update function & $\Tilde{\pmb\theta}_k$ & bootstrapping estimation for $\widehat{\pmb\theta}_k$  \\
\hline
\end{tabular}
\caption{Summary of key variables and their descriptions}
\label{table:variables}
\end{table}

At any time $t$, for each single cell in the $i$-th metabolic phase, denoted by $Z_t=i$ with $i\in\{0,1,\ldots, I\}$ (i.e., growth, stationary, death phases, etc), the dynamic behaviors of its metabolic network can be characterized by a stochastic model specified by parameters $\pmb{\alpha}_i$ and
a metabolic phase shift model with $\mathbb{P}_{ij}(\pmb{s}_t; \pmb{\beta})$ representing the transition probability from $Z_t=i$ to $Z_t=j$.  
Therefore, in this paper, we focus on calibrating the mechanistic model of the Bio-SoS specified by the calibration parameters $\pmb{\theta} \equiv \{\pmb{\alpha}_i, \pmb{\beta}\}$: (1) $\pmb{\alpha}_i$ characterizing cell dynamics in $i$-th metabolic phase or class; and (2) $\pmb{\beta}$ characterizing the phase shift probability that impacts the percentage of cells in each phase or class.
\begin{figure}[htbp]
    \centering
    \includegraphics[width=0.9\linewidth]{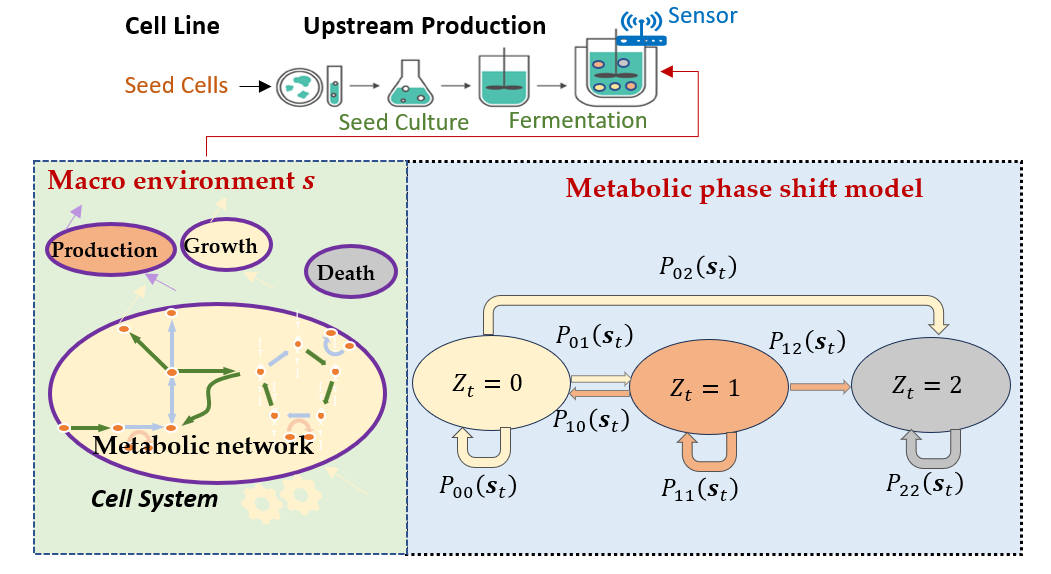}
    \caption{An illustration of the multi-scale mechanistic model for cell culture process and Bio-SoS.}
    \label{figure:Bio-SoS}
\end{figure}

\subsection{Single Cell Model}

This section describes a stochastic metabolic model to understand the dynamics of a single cell across different metabolic phases: growth ($Z=0$), production ($Z=1$) and so on. The model includes two main components: the Stochastic Modular Reaction Networks (SMRN) and the metabolic phase shift dynamics. Following the study \cite{anderson2011continuous}, the SMRN can be constructed $\Tilde{\mathbf{s}}^i_{t+1}=\Tilde{\mathbf{s}}^i_t+\mathbf {N} \cdot \pmb R_s^i(t)$,
where $\Tilde{\mathbf{s}}^i_t$ is the state for cells in $i$-th phase, $\mathbf{N}$ represents the stoichiometric matrix characterizing the structure of molecular reaction network, $\pmb R_s^i(t)$
follows a \textit{multivariate Poisson process} with molecular reaction rate vector $\pmb{\nu}(
\pmb{s}_t ;
\pmb{\alpha}_i)$. This reaction rate depends on macro-environment $\pmb{s}_t$, typically modeled using a Michaelis-Menten (MM) equation \citep{Sarantos2018,Keqi2024}.

The metabolic shift model 
addresses the transitions between metabolic phases based on environmental changes. At any time $t$, we have the metabolic phase-shift probability matrix $\mathbb{P}(\pmb{s}_t ;\pmb{\beta})$ depending on the environmental condition $\pmb{s}_t $, where each  element represents
\begin{equation}\nonumber
\mathbb{P}_{ii'}(\pmb{s}_t ) \equiv \mathbb{P}[Z_{t+1}=i'|Z_t=i]
~~~ \mbox{for} ~~~ i,i'=0,1,\ldots, I.
\end{equation}

\subsection{Macro-Kinetic State Transition} \label{subsec:macrometabolicShfit}


Suppose the environmental condition $\pmb{s}_t$ in bioreactor is homogeneous over space.  
The cell population dynamics are characterized by the evolution of cell densities in each metabolic phase $i$, i.e.
\begin{equation}\label{eq: cell growth model}
    \frac{\text{d} X_{i,t}}{\text{d} t} =\mu_{i}{(\pmb s_t )}X_{i,t}+ \sum_{i' \neq i}  \mathbb{P}_{i'i,t}({\pmb s_t }) X_{i',t},
\end{equation}
where $\mu_{i}{(\pmb s_t )}$ represents cell growth rate and the second term represents the instantaneous cells metabolic shifting from $i'$-th phase to $i$-th phase. Let $\pmb{s}_t$ denote metabolite concentrations (i.e., the number of molecules per unit of volume) in the system at time $t$. Define $\Delta \tilde{\mathbf{s}}_{t+1}^{(i,n)}$ as the change in metabolic concentration for the $n$-th cell in phase $i$, calculated as $\Tilde{\mathbf{s}}^{(i,n)}_{t+1} - \Tilde{\mathbf{s}}^{(i,n)}_t$. Then at any time $t$, given the density $X_{i,t}$ of cells in phase $i$, the overall change of metabolite concentration is the sum of the contributions from individual cells,
\begin{equation}\label{eq: metabolic network model}
    \pmb{s}_{t+1} - \pmb{s} _{t} =\sum_{i=0}^{I} \sum^{X_{i,t}}_{n=1} \Delta\tilde{\mathbf{s}}^{(i,n)}_{t+1}=\sum_{i=0}^{I}\mathbf {N}\pmb{R}^{i}(t),
\end{equation}
where $\pmb{R}^{i}(t)$ is a multivariate Poisson process with
reaction rate vector $X_{i,t}\pmb{v}(\pmb{s}_t;\pmb{\alpha}_i)$; as detailed in Eq.~(10) from \cite{Zheng2024}. Without losing generality, in the following sections, we consider a two-phase mechanistic model with parameters  $\pmb\theta=[\pmb\alpha_0,\pmb\alpha_1,\pmb\beta]$. 

\subsection{State Transition Probability}\label{subsec: approximated model}
In a short time interval $[t, t+\mathrm{d}t]$,  the change of state $\pmb{s}_t$ is denoted by
$    \mathrm{d} \pmb{s}_t=\pmb{s}_{t+\text{d}t}-\pmb{s}_t =\sum_{i=0}^{1}\mathbf {N}\text{d}\pmb{R}^{i}(t)=\mathbf {N}\text{d}\pmb{R}(t),$
where $\pmb{R}(t)$ is a multivariate Poisson process with parameter 
$\sum_{i=0}^{1}X_{i,t}\pmb\nu( \pmb{s}_t ;\pmb\alpha_i)$. Then $\E (\text{d}\pmb{R}\mid \pmb {s}_t)=\sum_{i=0}^{1}X_{i,t}\pmb\nu(\pmb{s}_t ;\pmb\alpha_i)$ and $\Var (\text{d}\pmb{R}\mid \pmb {s}_t)\approx \diag\{\sum_{i=0}^{1}X_{i,t}\pmb\nu^r(\pmb{s}_t ;\pmb\alpha_i)\}$, with $r=1,2,\cdots,R$ representing the component indices of $\pmb\nu$. The approximation is made under the assumption that the time interval is small, so the state doesn't change too much and different reactions can be treated independent. By forming as a Stochastic Differential Equations (SDEs), we have
\begin{align}\nonumber
\mathbf {N} \text{d}\pmb{R}(t)\mid\pmb s_t
&=\mathbf {N}\E (\text{d}\pmb{R}\mid \pmb {s}_t)\text{d}t+\mathbf {N}\Var (\text{d}\pmb{R}\mid \pmb {s}_t)\text{d}\pmb{B}_t\\\label{eq_fccccc}
&\approx \mathbf {N}\sum_{i=0}^{1}X_{i,t}\pmb\nu( \pmb{s}_t ;\pmb\alpha_i)\text{d}t+\left\{\mathbf {N} \diag\{\sum_{i=0}^{1}X_{i,t}\pmb\nu^r( \pmb{s}_t ;\pmb\alpha_i)\}\mathbf {N}^\top\right\}^{\frac{1}{2}}\text{d}\pmb{B}_t,
\end{align}
where $\pmb{B}_t$ denotes standard Brownian motion. 
{By Euler-Maruyama method \citep{Bayram2018}, a normal approximation can be applied to (\ref{eq_fccccc})}
\begin{align}\nonumber
    \mathbf {N}\Delta\pmb{R}(t)\mid \pmb s _t&\approx\mathbf {N}\sum_{i=0}^{1}X_{i,t}\pmb\nu(\pmb{s}_t ;\pmb\alpha_i)\Delta t+\left\{\mathbf {N} \diag\{\sum_{i=0}^{1}X_{i,t}\pmb\nu^r(\pmb{s}_t ;\pmb\alpha_i)\}\mathbf {N}^\top\right\}^{\frac{1}{2}}\Delta\pmb{B}_t\\\nonumber
    &\sim \mathcal{N}\left\{\mathbf {N}\sum_{i=0}^{1}X_{i,t}\pmb\nu(\pmb{s}_t ;\pmb\alpha_i)\Delta t, \mathbf {N} \diag\{\sum_{i=0}^{1}X_{i,t}\pmb\nu^r(\pmb{s}_t ;\pmb\alpha_i)\}\mathbf {N}^\top \diag\{\Delta t\}\right\},
\end{align}
where $\Delta\pmb{B}_t$ is a Gaussian random vector with mean zero and covariance matrix $\diag\{\Delta t\}$. 
Finally, we approximate conditional distribution of $\pmb{s}_{t+1}$ as
\begin{align}\label{eq: transition probability}
(\pmb{s}_{t+1}-
\pmb {s}_{t})|\pmb {s}_{t},\pmb{\theta}
\sim   \mathcal{N}\{\pmb{\mu}(\pmb\theta), \Sigma(\pmb\theta)\}, 
\end{align}
where $\pmb{\mu}(\pmb\theta)=\mathbf {N}\sum_{i=0}^{1}X_{i,t}\pmb\nu(\pmb{s}_t ;\pmb\alpha_i)\Delta t$, $ \Sigma(\pmb\theta)= \mathbf {N}\diag\{\sum_{i=0}^{1}X_{i,t}\pmb\nu^r(\pmb{s}_t ;\pmb\alpha_i)\}\mathbf {N}^\top \diag\{\Delta t\}\}$. 

\section{Digital Twin Calibration for Bio-SoS Mechanistic Model}
\label{sec:calibration}

To facilitate digital twin calibration, in this paper, we consider the physical system as a finite horizon stochastic process. Its dynamics can be characterized by a Bio-SoS mechanistic model specified by Eq.~\eqref{eq: cell growth model} and \eqref{eq: metabolic network model} with underlying true parameters, denoted by $\pmb{\theta}^c=[\pmb\alpha_0^c,\pmb\alpha_1^c,\pmb\beta^c]$.
Our goal is to develop a calibration method that can efficiently guide the DoEs and most informative data collection to improve model fidelity and the prediction accuracy of the digital twin.
For simplification, we consider batch-based cell culture experiments with the selection of initial concentration $\pmb{s}_0$ as the decision variable.

The trajectory $\tau=(\pmb{s}_0,\pmb{s}_1,\ldots, \pmb{s}_T)$ over a horizon of $T$ time steps has the joint probability density function 
     $p(\pmb{\tau}|\pmb{\theta},\pmb{s}_0) = \prod_{t=0}^{T-1}p(\pmb{s}_{t+1}|\pmb{s}_t;\pmb{\theta})$ with $\pmb{s}_0=\pi(\pmb\omega)$,
the calibration policy parameterized by $\pmb\omega$. Each transition probability $p(\pmb{s}_{t+1}|\pmb{s}_t;\pmb{\theta})$ follows a normal distribution as defined in Eq.~\eqref{eq: transition probability}. 
At each $k$-th calibration iteration, we select one design following the latest policy, i.e. $\pmb s_{0}^{(k)} = \pi(\pmb\omega_k)$,
run an experiment to collect one sample path from the physical system  and update the parameter $\widehat{\pmb{\theta}}_k$. Specifically, given the historical samples $\mathcal{D}_k=\{\pmb\tau_n\}_{n=1}^{k}$,
we consider maximizing the log-likelihood for $\pmb{\theta}$, 
\begin{equation}\label{MLE}
\widehat{\pmb{\theta}}_k \equiv \argmax_{\pmb{\theta}} L(\mathcal{D}_k; \pmb{\theta})= \argmax_{\pmb\theta}\sum_{n=1}^{k} \ell\left(\pmb\tau_n|\pmb\theta,\pmb{s}_{0}^{(n)}\right),
\end{equation}
where $L(\mathcal{D}_k,\pmb{\theta})$ denotes the sum of log-likelihoods and $\ell(\tau_n|\pmb\theta,\pmb{s}_0^{(n)})=\log p(\pmb{\tau}_n|\pmb{\theta},\pmb{s}_0^{(n)})$ is the log-likelihood of trajectory $\tau_n$ with initial state $\pmb{s}_0^{(n)}=\pi(\pmb{\omega}_n)$. The detailed derivation can be found in Section \ref{sec:MLE}.

For digital twin calibration, we select the new design of experiment 
at the beginning of $k$-th iteration. This selection is based on the evaluation of our current model on a  
set of prediction points $\mathcal{D}_{test}=\{\pmb s_0^h\}_{h=1}^H$ {sampled from a pre-specified distribution.}  
The objective is to minimize the MSE of digital twin prediction, 
\begin{equation}\label{eq: policy optimization}
\pmb\omega^\star_k  \triangleq \argmin_{\pmb\omega} J_k(\pmb\omega) \quad 
\mbox{with}
\quad
J_k(\pmb\omega)  \equiv\sum_{h=1}^H\mathbb{E}\left[
    Y^{\text{p}}\left(\pmb s^{h}_{0},\pmb{\theta}^c\right) - Y^{\text{d}}\left(\pmb s^{h}_{0},{\pmb{u}}_{k}(\widehat{\pmb\theta}_{k-1},\pmb\omega)\right) 
    \right]^2,
\end{equation}
where $\pmb{u}_k(\widehat{\pmb\theta}_{k-1},\pmb\omega)$ denote the virtual model update function with the simulated trajectory $\pmb\tau^{\text{d}}_k$ (see Eq.~\eqref{eq: virtual}), $Y=\phi(\pmb{s}_T)$ represents the desired system output, such as yield and product quality attributes. The superscript "p" represents for physical system and "d" for digital twin. 

To sample efficiently guide digital twin calibration, we propose a gradient-based optimal learning approach with the procedure illustration 
as shown in Figure~\ref{fig:proposed calibration approach} and develop a calibration algorithm, {which is summarized as Algorithm~\ref{alg:cap}.}
In specific, for optimization problem \eqref{MLE}, we use gradient-based method  
 \begin{equation}
\widehat{\pmb\theta}_{k}\leftarrow \widehat{\pmb\theta}_{k-1}+\text{Adam}_k(\nabla_{\pmb{\theta}} \ell(\pmb\tau_k; \widehat{\pmb{\theta}}_{k-1},\pmb{s}_0^{(k)})),     
 \end{equation} where ``Adam'' represents Adam optimizer \citep{Kingma2014AdamAM}.
The details of the optimization steps will be discussed in Section \ref{sec:MLE}. Now we can define  virtual model update function
\begin{equation}\label{eq: virtual}
\pmb{u}_k(\pmb\theta,\pmb\omega)=\pmb\theta + \text{Adam}_k\left(\nabla_{\pmb\theta}\ell(\pmb\tau_{k}^{\text{d}};\pmb\theta,\pmb\omega)\right)   
\end{equation} with a simulated trajectory $\pmb\tau^{\text{d}}_k \sim p(\tau|\widehat{\pmb\theta}_{k-1},\pmb{s}_0)$ and $\pmb{s}_0=\pi(\pmb\omega)$.
Then to guide informative experiment and improve digital twin prediction, we solve the optimization~\eqref{eq: policy optimization} by using the stochastic gradient descent
\begin{equation}\label{gra_update}
    \pmb\omega_k=\pmb\omega_{k-1}-\gamma_k\nabla_{\pmb\omega} J_k(\pmb\omega)\mid_{\pmb\omega=\pmb\omega_{k-1}}.
\end{equation}
\begin{figure}[H]
        \centering
        \includegraphics[width=0.8\linewidth]{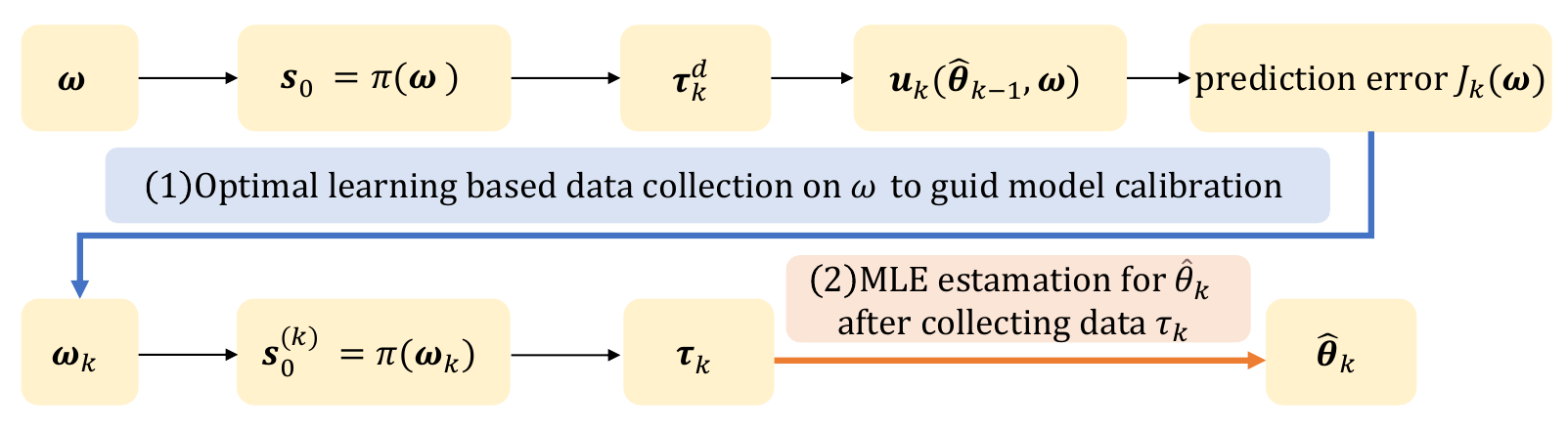} 
        \vspace{-0.05in}
        \caption{The procedure illustration of the proposed calibration approach.}
        \label{fig:proposed calibration approach} \vspace{-0.1in}
    \end{figure}   



Inspired by Figure~\ref{fig:proposed calibration approach}, we compute the gradient with respect to $\pmb{\omega}$ using the chain rule. Specifically, we take the gradient with respect to $ \pmb{u}_k$, and then take the gradient to $ \pmb{u}_k$ with respect to $\pmb{\omega}$, i.e.,
\begin{equation}
    \pmb\omega_k
=\pmb\omega_{k-1}-\gamma_k\nabla_{\pmb{u}_k} \sum_{h=1}^H\mathbb{E}\left[
    Y^{\text{p}}\left(\pmb s^{h}_{0},\pmb{\theta}^c\right) - Y^{\text{d}}\left(\pmb s^{h}_{0},\pmb{u}_k\right) 
    \right]^2 \Bigm|_{\pmb u_k=(\widehat{\pmb\theta}_{k-1},\pmb\omega_{k-1})}\times\nabla_{\pmb{\omega}} \pmb{u}_k\left(\widehat{\pmb\theta}_{k-1},\pmb\omega\right)\Bigm|_{\pmb\omega=\pmb\omega_{k-1}}.
\end{equation}
Since $\pmb\theta^c$ is unknown, we use the bootstrapping to quantify the model parameter estimation uncertainty
in Section \ref{sec:optimal learning}.
Following the upper part of Figure~\ref{fig:proposed calibration approach}, the prediction error propagates from $\pmb{\omega}$ to model parameter update $\pmb{u_k}$, and  then to our objective $J_k(\pmb\omega)$.  By taking the gradient of $J_k(\pmb\omega)$, we can discern the influence of $\pmb{\omega}$ 
on the final prediction error, which assists in selecting a more effective design. Since deriving an exact expression for this gradient is complex, we develop an approximate closed-form representation, which will be detailed in Section~\ref{sec:optimal learning}. This approximation provides an efficient method for improving calibration policy.

\section{Model Inference and Parameter Update}
\label{sec:MLE}
The gradient of the log-likelihood is calculated as
\begin{align}\nonumber
    \nabla_{\pmb{\theta}} L(\mathcal{D}_k; \pmb{\theta})
     =\nabla_{\pmb{\theta}}\left[ \sum_{n=1}^{k}\log  p(\pmb{\tau}_k|\pmb{\theta},\pmb s_0^{(n)})\right]
    =\sum_{n=1}^{k}\sum_{t=0}^{T-1}\left[ \nabla_{\pmb{\theta}}\log  p(\pmb{s}^{(n)}_{t+1}|\pmb{s}^{(n)}_t ,\pmb{\theta})\right].
\end{align} Each term \(\log p(\pmb{s}_{t+1}^{(n)} | \pmb{s}_t^{(n)}, \pmb{\theta})\) in this summation involves gradients that are detailed in the following lemma.
\begin{lemma}
Conditioning on the current state $\pmb {s}_t$ and model parameter $\pmb\theta$, the state change during in any small time interval $(t, t + \Delta t]$ is $\Delta \pmb{s}_{t+1}=\pmb{s}_{t+1}-\pmb {s}_t$ following a multi-variate normal distribution $\mathcal{N}\{\pmb{\mu}(\pmb\theta), \Sigma(\pmb\theta)\}$. Then the gradient of log-likelihood $\log p(\pmb{s}_{t+1} | \pmb{s}_t, \pmb{\theta})$ with respect to $\pmb{\theta}$ is given by
\begin{align*}
 \nabla_{\pmb\theta} \log p(\pmb{s}_{t+1}| \pmb {s}_t,\pmb\theta)=&-\frac{1}{2}\text{tr}\left(\Sigma(\pmb\theta)^{-1} \nabla_\theta\Sigma(\pmb\theta)\right)+\nabla_\theta{\pmb\mu}(\pmb\theta)^\top \Sigma(\pmb\theta)^{-1} (\pmb{s}_{t+1}-\pmb {s}_t - \pmb\mu(\pmb\theta))\\  &-\frac{1}{2} (\pmb{s}_{t+1}-\pmb {s}_t - \pmb\mu(\pmb\theta))^\top \left(-\Sigma(\pmb\theta)^{-1} {\nabla_{\pmb\theta}\Sigma(\pmb\theta)}\Sigma(\pmb\theta)^{-1}\right)(\pmb{s}_{t+1}-\pmb {s}_t - \pmb\mu(\pmb\theta)) . 
\end{align*}    
\end{lemma}
\begin{proof}
By the definition of multi-variant normal distribution, we have
\[
p(\pmb{s}_{t+1}| \pmb{s}_t,\pmb\theta) = \frac{1}{(2\pi)^{M/2}|\Sigma(\pmb\theta)|^{1/2}} \exp\left(-\frac{1}{2}(\pmb{s}_{t+1}-\pmb {s}_t-\pmb\mu(\pmb\theta))^\top\Sigma(\pmb\theta)^{-1}(\pmb{s}_{t+1}-\pmb {s}_t-\pmb\mu(\pmb\theta))\right).
\]
Taking the logarithm on both side, we have
\[
\log p(\pmb{s}_{t+1}| \pmb {s}_t,\pmb\theta)  = -\frac{1}{2} \log |\Sigma(\pmb\theta)| - \frac{1}{2}  (\pmb{s}_{t+1}-\pmb {s}_t - \pmb\mu(\pmb\theta))^\top \Sigma(\pmb\theta)^{-1} (\pmb{s}_{t+1}-\pmb {s}_t - \pmb\mu(\pmb\theta))
-\frac{M}{2} \log (2\pi).\]
The gradient for the first term is
$\nabla_{\pmb\theta} \log |\Sigma(\pmb\theta)| = \text{tr}\left(\Sigma(\pmb\theta)^{-1} \nabla_{\pmb\theta}\Sigma(\pmb\theta)\right), $
where $\nabla_{\pmb\theta}\Sigma(\pmb\theta)$ is a tensor of $M\times M\times q$ dimension with $q$ as the length of $\pmb\theta$, and tr$(\cdot)$ represents the trace for first two dimensions. The gradient for the second term is
\begin{align*}
 {\nabla_{\pmb\theta}}(\pmb{s}_{t+1}-\pmb {s}_t - \pmb\mu(\pmb\theta))^\top \Sigma(\pmb\theta)^{-1} (\pmb{s}_{t+1} - \pmb{s}_t - \pmb\mu(\pmb\theta)) = -2 \nabla_{\pmb\theta}{\pmb\mu}(\pmb\theta)^\top \Sigma(\pmb\theta)^{-1} (\pmb{s}_{t+1}-\pmb {s}_t - \pmb\mu(\pmb\theta)) \\
 + (\pmb{s}_{t+1}-\pmb {s}_t - \pmb\mu(\pmb\theta))^\top \left(-\Sigma(\pmb\theta)^{-1} {\nabla_{\pmb\theta}\Sigma(\pmb\theta)}\Sigma(\pmb\theta)^{-1}\right)(\pmb{s}_{t+1}-\pmb {s}_t - \pmb\mu(\pmb\theta)).\tag*{\qedhere}
\end{align*}
\end{proof}


 Given the calculated gradient, we can use the Adam optimizer \citep{Kingma2014AdamAM} for updating $\widehat{\pmb\theta}$. Following the algorithm in their paper, for the $k$-th iteration, we have

\begin{equation}\label{eq: param_update}
\widehat{\pmb\theta}_k = \widehat{\pmb\theta}_{k-1}+\lambda \cdot \frac{\left(1-\xi_1\right)\sum_{l=1}^{k}\xi_1 ^
{k-l}G_l}{\left(1-\xi_1^k\right)\sqrt{\frac{\left(1-\xi_2\right)\sum_{l=1}^{k}\xi_2 ^
{k-l}G_l^2}{\left(1-\xi_2^k\right)}+\epsilon}},    
\end{equation}
where $\lambda$ is the learning rate, $G_l=\nabla_{\pmb{\theta}} \ell(\pmb\tau_l; \widehat{\pmb{\theta}}_{l-1},\pmb{s}_0^{(l)})$ is the gradient of the likelihood function at $l$-th iteration, $g_l^2$ is the element-wise
square, $\xi_1=0.9, \xi_2=0.999$ are hyper-parameters related to  exponential decay rates, and $\epsilon=10^{-8}$ is a small number to avoid zero in the denominator.

\section{Policy Gradient Estimation and Optimal Learning}
\label{sec:optimal learning}

In this section, we turn our focus on learning the calibration policy by minimizing the prediction MSE in \eqref{eq: policy optimization}. To develop a surrogate model of this MSE, we employ the {LNA}  in Section~\ref{sec: LNA} to approximate the stochastic dynamics of the system, reflecting its sensitivity to initial conditions and parameter values. This approach allows us to describe the system state at any given time as a normally distributed random variable, facilitating the derivation of analytical expressions for the gradient estimator. To keep the mechanism information, in Section \ref{ode}, a first-order Euler's Method is utilized to derive the closed-form solutions of the SDE-based mechanisms 
obtained from {LNA}.

\subsection{Linear Noise Approximation on Bio-SoS Dynamics}
\label{sec: LNA}
We use {LNA} \citep{Paul2014} to derive the solution of SDEs in {Lemma~\ref{lemma:LNA}}, which allows us to analyze the estimation error propagation from ${\pmb{\theta}}$ to the output prediction and develop a surrogate model of the MSE \eqref{eq: policy optimization}. Conditional on $\pmb s_0$ and $\pmb\theta$, we denote the state of systems at time $t$ as $\pmb s_t(\pmb\theta)$. 
Considering the SDEs characterizing the Bio-SoS dynamics as shown in  \eqref{eq_fccccc}, we have
\begin{align}\label{SDEsys}
    \mathrm{d} \pmb{s}_t(\pmb{\theta})&=\pmb{b}(\pmb{s}_t(\pmb{\theta}))\text{d}t+ \mathbf{Q}(\pmb{s}_t(\pmb{\theta}))\text{d}\pmb{B}_t,
\end{align}
where
$\pmb{b}=\mathbf{N}\sum_{i=0}^{1}X_{i,t+\text{d}t}\pmb\nu(\pmb{s}_t ;\pmb\alpha_i)$ and $
    \mathbf{Q}=\left\{\mathbf{N}\diag\{\sum_{i=0}^{1}X_{i,t+\text{d}t}\pmb\nu^k(\pmb{s}_t ;\pmb\alpha_i)\}\mathbf{N}^\top\right\}^{\frac{1}{2}}.$\\

\begin{lemma}\label{lemma:LNA}
(Linear Noise Approximation) Suppose $\pmb{s}_0 \sim N\left(\boldsymbol{\mu}_0^*, \boldsymbol{\Sigma}_0^*\right)$,  $\boldsymbol{\eta}_0=\boldsymbol{\mu}_0^*$ and $\boldsymbol{\Psi}_0=\boldsymbol{\Sigma}_0^*$.
Then the solution of SDE \eqref{SDEsys} can be approximated as $$
\pmb{s}_t(\pmb\theta) \sim N\left(\boldsymbol{\eta}_t, \boldsymbol{\Psi}_t\right)
,$$
where $\boldsymbol{\eta}_t, \boldsymbol{\Psi}_t$ are the solution of ordinary differential equations (ODEs)
\begin{equation}\label{ode1}
    \frac{\text{d}\boldsymbol{\eta}}{\text{d}t}=\pmb{b}(\boldsymbol{\eta}),\quad \frac{\mathrm{d} \boldsymbol{\Psi}}{\mathrm{d} t}  =\boldsymbol{\Psi} \mathbf{F}^\top+\mathbf{F} \boldsymbol{\Psi}+\mathbf{Q Q}^\top,
\end{equation}
where we assume $\pmb{s}_t=\boldsymbol{\eta}_t+\mathbf{M}_t$
with $\boldsymbol{\eta}_t$ representing the deterministic path and $\mathbf{M}_t$ representing stochastic path, $\boldsymbol{\Psi}_t\triangleq\operatorname{Var}[\mathbf{M}_t], $ $\mathbf{F}_t$ is the $M \times M$ matrix with components
$$
\mathbf{F}_{t,i j}=\left.\frac{\partial \pmb{b}_i}{\partial \pmb s_{t,j}}\right|_{\boldsymbol{\eta}_t} \text {   and   } \mathbf{Q}_t=\mathbf{Q}(\boldsymbol{\eta}_t) .
$$
\end{lemma} 
\begin{proof}
Proof can be find in \cite{Paul2014} with $\mathbf{m}_t\triangleq\mathbb{E}[\mathbf{M}_t]=0$.    
\end{proof}




Obtaining an analytical solution for the series of ODEs described in \eqref{ode1} is challenging. Therefore, we use a numerical solution approach, as described in Section \ref{ode}. This method involves applying the solution iteratively at each time step $t$, allowing us to approximate the system's behavior over time. For the physical system and digital twin, we denote 
$
\pmb{s}_t(\pmb{\theta}^c) \sim N\left(\boldsymbol{\eta}_t,\boldsymbol{\Psi}_t\right)$ and $
\pmb{s}_t({\pmb{u}}_k) \sim N\left(\widehat{\boldsymbol{\eta}}_t, \widehat{\boldsymbol{\Psi}}_t\right).
$
Then we have
$
\pmb{s}_T(\pmb\theta^c)-\mathbf{s}_T({\pmb{u}}_k)\sim N\left(\boldsymbol{\eta}_T-\widehat{\boldsymbol{\eta}}_T, \boldsymbol{\Psi}_T+\widehat{\boldsymbol{\Psi}}_T\right).
$ For any $m$-th component $Y$, the surrogate model of the MSE \eqref{eq: policy optimization} is given by
\begin{align}
    \E\left[\left(Y^{\text{p}}(\pmb{\theta}^c)-Y^{\text{d}}({\pmb{u}}_k)\right)^2\right]
    =\left[\boldsymbol{\eta}_{T,m}-\widehat{\boldsymbol{\eta}}_{T,m}\right]^2+\boldsymbol{\Psi}_{T,mm}+\widehat{\boldsymbol{\Psi}}_{T,mm}.\label{eq_ex_de} 
\end{align}
\subsection{First-Order Euler's Method for Solving LNA}\label{ode}

To get the solution of $\pmb{s}_t$, we solve the ODEs outlined in \eqref{ode1} numerically by using first-order Euler's method. 
For the term $\boldsymbol{\eta}_t$, the update equation is given by
\begin{align}\label{eq_de_1}
\boldsymbol{\eta_T}&=\boldsymbol{\eta}_{T-1}+b(\boldsymbol{\eta_{T-1}})\Delta t=\boldsymbol{\eta}_{0}+\sum_{t=0}^{T-1}b(\boldsymbol{\eta_{t}})\Delta t=\boldsymbol{\eta}_{0}+\sum_{t=0}^{T-1}\sum_{i=0}^{1}X_{i,t}\mathbf{N}\pmb\nu(\pmb{s}_t ;\pmb\alpha_i)\Delta t\triangleq\boldsymbol{\eta}_{0}+\sum_{t=0}^{T-1}f_t{(\pmb{\theta})}.
\end{align}
For the term $\boldsymbol{\Psi}_{t}$, the update equation is given by
\begin{align}\nonumber
   \boldsymbol{\Psi}_T &=\boldsymbol{\Psi}_{T-1}+\left[\boldsymbol{\Psi}_{0} \mathbf{F}_{T-1}^\top+\mathbf{F}_{T-1} \boldsymbol{\Psi}_0+\mathbf{Q_{T-1} Q_{T-1}}^\top\right]\Delta t+O((\Delta t)^2)\\\nonumber
&=\boldsymbol{\Psi}_{0}+\sum_{t=0}^{T-1}\left[\boldsymbol{\Psi}_{0} \mathbf{F}_{t}^\top+\mathbf{F}_{t} \boldsymbol{\Psi}_0+\mathbf{Q_{t} Q_{t}}^\top\right]\Delta t+O((\Delta t)^2)\\\nonumber
   &=\boldsymbol{\Psi}_{0}
+\boldsymbol{\Psi}_{0} \left[\sum_{t=0}^{T-1}\sum_{i=0}^{1}\left[\mathbf{N}\pmb\nu(\pmb{s}_t ;\pmb\alpha_i)\nabla_{\pmb{s}_t } X_{i,t+\text{d}t}\right]\right]^\top\Delta t+\left[\sum_{t=0}^{T-1}\sum_{i=0}^{1}\left[\mathbf{N}\pmb\nu(\pmb{s}_t ;\pmb\alpha_i)\nabla_{\pmb{s}_t } X_{i,t+\text{d}t}\right]\right] \boldsymbol{\Psi}_0\Delta t\\\nonumber
   &~~~~+\boldsymbol{\Psi}_{0} \left[\sum_{t=0}^{T-1}\sum_{i=0}^{1}\left[ X_{i,t+\text{d}t}\mathbf{N}\nabla_{\pmb{s}_t }\pmb\nu(\pmb{s}_t ;\pmb\alpha_i)\right]\right]^\top\Delta t+\left[\sum_{t=0}^{T-1}\sum_{i=0}^{1}\left[X_{i,t+\text{d}t}\mathbf{N}\nabla_{\pmb{s}_t }\pmb\nu(\pmb{s}_t ;\pmb\alpha_i)\right]\right] \boldsymbol{\Psi}_0\Delta t\\\nonumber
   &~~~~+\sum_{t=0}^{T-1}\left[\mathbf{N}\diag\{\sum_{i=0}^{1}X_{i,t+\text{d}t}\pmb\nu^r(\pmb{s}_t ;\pmb\alpha_i)\}\mathbf{N}^\top\right]\Delta t
    +O((\Delta t)^2)
    \\\label{eq_de_3} 
   & \triangleq\boldsymbol{\Psi}_{0}+\sum_{t=0}^{T-1}g_t(\pmb\theta)+O((\Delta t)^2).
\end{align}
Then by \eqref{eq_ex_de}-\eqref{eq_de_3}, the final output discrepancy can be written as 
\begin{align}\nonumber
    ~&\E\left[\left(Y^{\text{p}}(\pmb{\theta}^c)-Y^{\text{d}}({\pmb{u}}_k)\right)^2\right]
    =\left(\sum_{t=0}^{T-1}[f_{t,m}(\pmb\theta^c)-f_{t,m}({\pmb u}_k)]\right)^2+2\Psi_{0,mm}+\sum_{t=0}^{T-1}[g_{t,mm}(\pmb\theta^c)+g_{t,mm}({\pmb u}_k)].
\end{align}
Now conditional on $\widehat{\pmb{\theta}}_{k-1}$, by applying a chain rule, we can update the policy parameters by

\begin{align}\nonumber
\pmb\omega_k
\approx
&\pmb\omega_{k-1}-\gamma_k
\nabla_{{\pmb u}_k}\sum_{h=1}^H\left[\left(\sum_{t=0}^{T-1}[f_{t,m}(\pmb\theta^c)-f_{t,m}({\pmb u}_k)]\right)^2+\sum_{t=0}^{T-1}[g_{t,mm}(\pmb\theta^c)+g_{t,mm}({\pmb u}_k)]\middle]\right|_{\pmb u_k=(\widehat{\pmb\theta}_{k-1},\pmb\omega_{k-1})}
\\\label{policy_update}
&\times\lambda\nabla_{\pmb{\omega}}  \frac{\left(1-\xi_1\right)\sum_{t=1}^{k}\xi_1 ^
{k-t}\Tilde{G}_t}{\left(1-\xi_1^k\right)\sqrt{\frac{\left(1-\xi_2\right)\sum_{t=1}^{k}\xi_2 ^
{k-t}\Tilde{G}_t^2}{\left(1-\xi_2^k\right)}+\epsilon}}\Biggm|_{\pmb\omega=\pmb\omega_{k-1}}
,\end{align}
where $\Tilde{G}_k=\nabla_{\pmb{\theta}} \ell(\pmb\tau_k^{\text{d}}; \widehat{\pmb{\theta}}_{k-1},\pmb{s}_0^{(k)})$  and $\Tilde{G}_t=G_t, \text{for } t=1,\ldots,k-1$. Here $\Tilde{G}_k$ is a gradient based on simulation trajectory $\pmb\tau_k^{\text{d}}$ since  $\pmb\omega_k$ is updated at the beginning of $k$-th iteration and we can not get physical data $\pmb\tau_k$.

For the estimation of $\pmb\theta^c$, we use the parametric bootstrapping. Based on MLE and data $\pmb\tau_1,\ldots,\pmb\tau_{k-1}$, we have $\widehat{\pmb\theta}_{k-1}$ as the estimator for $\pmb\theta^c$. Then we can generate $l$ bootstrap samples,
denoted as 
$\Tilde{\pmb\theta}^1_{k-1}$, $\Tilde{\pmb\theta}^2_{k-1}, \ldots,$ $\Tilde{\pmb\theta}^l_{k-1}$.
Then we can take the estimation
$    f_{t,m}(\pmb\theta^c)\approx\frac{1}{l}\sum_{n=1}^l f_{t,m}(\Tilde{\pmb\theta}_{k-1}^n),\quad g_{t,mm}(\pmb\theta^c)\approx \frac{1}{l}\sum_{n=1}^l g_{t,mm}(\Tilde{\pmb\theta}_{k-1}^n).$
{Finally, we summarize the algorithm as follows.} \\
\begin{algorithm}[H]
\caption{Gradient-based Optimal Learning}\label{alg:cap}
\textbf{Input}:initial policy parameter $\pmb\omega_0$, initial model parameter $\widehat{\pmb\theta}_0$, step size $\lambda$ for parameter update, initial step size $\gamma_0$ for design policy update\\
\textbf{For} $k=1:K$\\
\begin{enumerate}
      \item Compute $\pmb{u}_k(\widehat{\pmb\theta}_{k-1} ,\pmb\omega)$ by Eq.~\eqref{eq: virtual} \\ 
      \item Update design  $\pmb\omega_k$ by Eq.~\eqref{policy_update}
    \item Do experiments in physical system with design $\pi{(\pmb\omega_k)}$ to get new data $\pmb{\tau}_k$,     $\mathcal{D}_{k}= \mathcal{D}_{k-1}\cup\{\pmb{\tau}_k\}$
     \item Update model parameters $\widehat{\pmb\theta}_{k}$ by Eq.~\eqref{eq: param_update}
\end{enumerate}
\vspace{-5mm}
\textbf{End}
\end{algorithm}
\section{Empirical Study}\label{sec: empirical study}
{For validation, we use a simplified version of the Chinese Hamster Ovary (CHO) cell culture model \citep{Ghorbaniaghdam2014} as our \textbf{SMRN} model. We combine it with a state shift model to simulate cell culture dynamics. Our gradient-based policy is compared against a random policy and a GP-based calibration approach. Details of experiment setup are provided in Section \ref{sec_6_1}, with results in Section \ref{sec_6_2}.}

\subsection{Cell Culture Model and Validation}\label{sec_6_1}
CHO cells have become the most commonly used mammalian hosts for the industrial production of monoclonal antibodies (mAbs) and recombinant protein therapeutics. 
The {cell growth model} is described by a function characterizing different phases of cell culture
\begin{align*}
        v_{growth0}&= v_{maxg0}\frac{GLC}{K_{mGLC_0} + GLC} X_{0}, 
        \quad
     \frac{\mathrm{d} X_0}{\text{d}t}=v_{growth0}
     +\beta_1 X_{1}-\beta_0X_{0},\quad
     \frac{\mathrm{d} X_{1}}{\text{d}t}=\beta_0 X_{0}-\beta_1 X_{1},   
\end{align*}
where $v_{growth0}$ represents the growth rate that depends on glucose concentration in the growth phase (phase 0) and where $\beta_0$
  and $\beta_1$ are phase transition rates, reflecting the dynamic interplay between growth and production phases. A simplified {SMRN} focuses on three key metabolisms: glucose consumption, lactate production, and mAb synthesis with chemical equations $(1)~ \text{GLC}\rightarrow 2\text{LAC}$, $(2)~ 0.02\text{GLC}\rightarrow \text{X}$, and 
    $(3)~ 0.01\text{GLC}\rightarrow \text{mAb}$,
where GLC and LAC represent glucose and lactate. The state of the system at any time $t$ is denoted by $\pmb s_t=\left[[GLC]_t,[LAC]_t,[mAb]_t\right]$ representing the concentrations of glucose, lactate, and mAb, respectively. The flux rates $\pmb v(\pmb{s}_t;\pmb\alpha_i)= (-v_{HK},2v_{HK}, v_{mAb})$, where
  \begin{align}\nonumber
    v_{HK_i} &= v_{maxHK_i} \frac{GLC}{K_{mGLC_i} + GLC} X_{i,t}, \quad v_{HK}=v_{HK_0}+v_{HK_1}
     \\\nonumber    
     v_{mAb_i}&= v_{maxmAb_i}\frac{GLC}{K_{mGLC_i} + GLC} X_{i,t}, \quad v_{mAb}=v_{mAb_0}+v_{mAb_1}
\end{align}

The metabolic reactions are formulated as a SDE system~\eqref{SDEsys} to account for the inherent randomness and fluctuations in biochemical reactions. The general form of this SDE is given by
$\mathrm{d} \pmb{s}_t(\pmb{\theta})=\pmb{b}(\pmb{s}_t(\pmb{\theta}))\text{d}t+ \mathbf{Q}(\pmb{s}_t(\pmb{\theta}))\text{d}\pmb{B}_t,$
where the mean function
$\pmb{b}=\mathbf{N}\sum_{i=0}^{1}X_{i,t+\text{d}t}\pmb\nu(\pmb{s}_t ;\pmb\alpha_i)$ is modeled by the stoichiometric matrix $\mathbf{N} = [-1, -0.02, -0.01; ~2, 0, 0;~ 0, 0, 1]
$ and reaction rates $\pmb{v}(\pmb{s}_t;\pmb\alpha_i)$. The system variability is modeled by $
    \mathbf{Q}=\left\{\mathbf{N}\diag\{\sum_{i=0}^{1}X_{i,t+\text{d}t}\pmb\nu^r(\pmb{s}_t ;\pmb\alpha_i)\}\mathbf{N}^\top\right\}^{\frac{1}{2}}$.

{We fit our model by minimizing $\text{MSE}=\frac{1}{n}\sum_{j=1}^{n}{[Y^s_j-Y^e_j]^2}$ between experiment data $Y^e$ from \cite{Ghorbaniaghdam2014}
 and simulation $Y^s$ from our model.} The results of our model validation are depicted in Figure \ref{fig: model fitting}, where the measurements are from \cite{Ghorbaniaghdam2014}. 
 The confidence interval (CI) is calculated by the formula
$[\bar{x} - z {\sigma}/{\sqrt{n}},\bar{x} + z {\sigma}/{\sqrt{n}}]$,
where $\bar{x}$ is the sample mean, $z$ 
 is the z-score corresponding to the 95\% confidence level, which is approximately 1.96, $\sigma$ is the standard deviation of the sample, $n$ is the number of observations in the sample. This result shows the model’s effectiveness in capturing the dynamics of the CHO cell culture system.
 \begin{figure}[htbp]
    \centering
    \includegraphics[width=1\linewidth, height=0.2\linewidth]{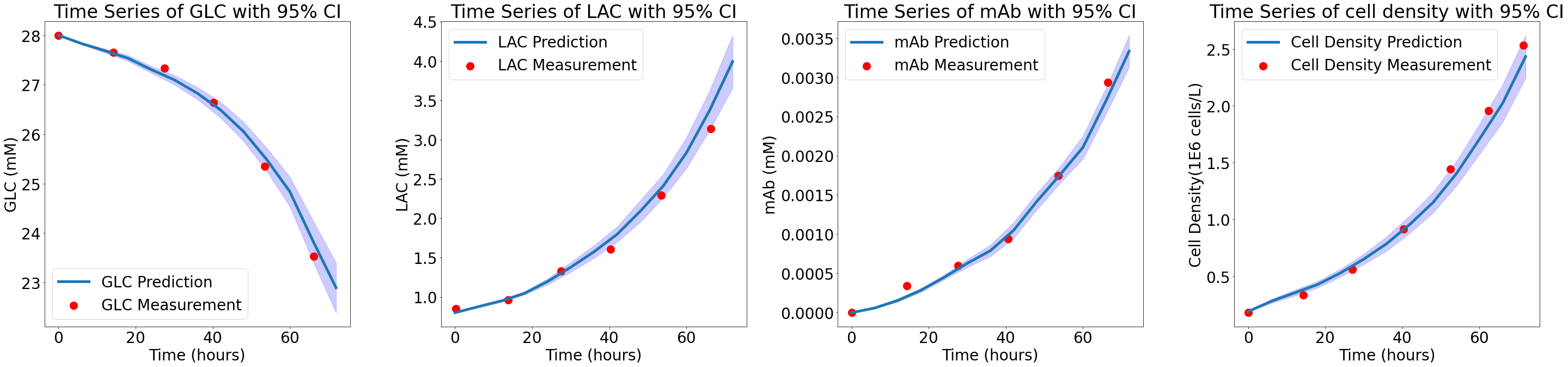}
    \caption{{The cell culture mechanistic model was validated by using the experimental data. Red dots represent the experimental data from the literature, while the blue lines indicate the simulation predictions from our model, with the mean and 95\% CI calculated from 20 replications.}}
    \label{fig: model fitting}
\end{figure}

 In the following calibration experiments, we use the fitted parameter as the true model parameter of the physical system and assume the digital twin has the same model structure with unknown parameters. The true parameters are (1) $\pmb\alpha_0=[v_{maxHK0},v_{maxmAb0},K_{mGLC_0},v_{maxg0}]=[4.15\times10^{-5},3.5\times10^{-8},4.2, 0.126] $; (2) $\pmb\alpha_1=[v_{maxHK1},v_{maxmAb1},K_{mGLC_1}]=[1.95\times10^{-5},7.5\times10^{-8},3.5] $ and (3) $\pmb\beta=[\beta_0,\beta_1]=[0.3,0.1]$.

\subsection{Digital Twin Calibration}\label{sec_6_2}
With the true parameters defined in last subsection,
our objective is to calibrate the digital twin model by strategically designing CHO cell culture experiments. The aim is to minimize the model prediction error by choosing the initial glucose concentration (design variable). In each experiment, we adjust the initial glucose concentration and then collect a new batch of data to continuously refine our model. 

For the calibration experiments setting, we assume time course data of the states is collected every 12h until 72h, so each batch includes 7-time sequence data. We focus on calibrating parameters $\pmb\alpha_0, \pmb\alpha_1$. The parameter {$\pmb\beta$ is not considered during the digital twin calibration because it's directly related to cell density, which is not built into our state.} The design space for the initial concentrations of glucose is set within the range 
$[18,38]$ mM, and the initial concentration of lactate and mAb are $[LAC]_0=0.8$ mM, $[mAb]_0=0$ mM, the initial cell densities in Phase 0 and 1 are set to be $X_{0,0}=180\times10^6$ cells/L, $X_{1,0}=
10\times 10^6$ cells/L, respectively. The initial parameters for the digital twin are estimated by the first batch of data using MLE. Since our parameters are in different scales, we set a learning rate {$\lambda=[3\times 10^{-5},3\times 10^{-7},3,0.3,3\times 10^{-5},3\times 10^{-7},0.3]$}  for updating the parameter and $\gamma_k =5$ (for all $k$) for updating next design policy. The fixed set of test data $\mathcal{D}_{test}$ is composed of initial states $\pmb s_0^1=[28,0.8,0]$ and $\pmb s_0^2=[38,0.8,0]$.

We conducted 15 macro-replications, each consisting of $K=200$ iterations under varying initial conditions, with initial glucose levels randomly chosen from the uniform distribution over the interval 
$[18,38]$.  We compare the proposed gradient-based calibration approach with (1) a random design approach, where the design is uniformly sampled from $[8, 38]$; {(2) GP-based approach with 
 expected improvement acquisition function and 20 initial points. A Matern Kernal $K_{\text{Matern}}\left(x, x^{\prime}\right)=\frac{2^{1-\nu}}{\Gamma(\nu)}\left(\frac{\sqrt{2 \nu}|d|}{\ell}\right)^\nu K_\nu\left(\frac{\sqrt{2 \nu}|d|}{\ell}\right)$ is used with $d=x-x'$, $\ell=1$ and $\nu=1.5$. The GP regression is fitted using the MSE as the output and parameters as input. The GP model is iteratively updated with data from new physical experiments.} To ensure consistency across the three approaches, we used the same random seeds for these three approaches. 
The mean prediction performance and the 95\% confidence interval, calculated across macro-replications, represented by the Mean Relative Errors $\text{MRE}=\frac{1}{7}\sum_{j=1}^{7}{|\widehat\theta_{(j)}-\theta_{(j)}^c|}/{|\theta_{(j)}^c|}$ with script $(j)$ representing the $j$-th component, are presented in Figure~\ref{fig:pre_1} for mAb protein drug generation prediction and in Figure \ref{fig:pra_1} for parameter estimation, respectively. For all figures, red lines represent the performance of our approach, green lines represent the random design and {blue lines represent GP.}


\begin{figure}[htbp]
    \centering
    \begin{minipage}{0.45\linewidth}
        \centering        \includegraphics[width=\linewidth]{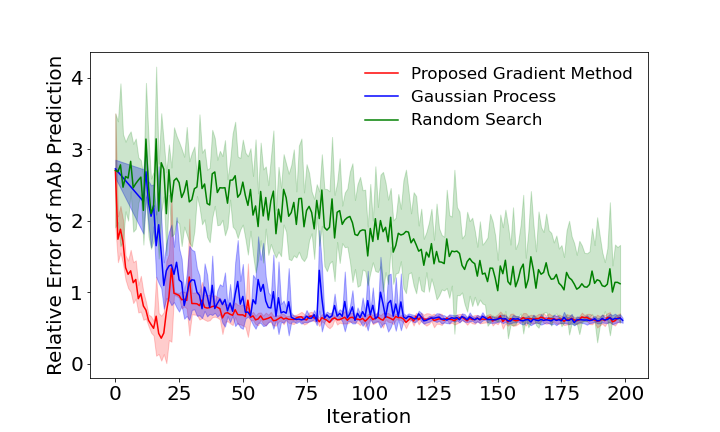}
        \caption{mAb Prediction Error}
        \label{fig:pre_1}
    \end{minipage}
    \hspace{0.01\linewidth} 
    \begin{minipage}{0.45\linewidth}
        \centering
        \includegraphics[width=\linewidth]{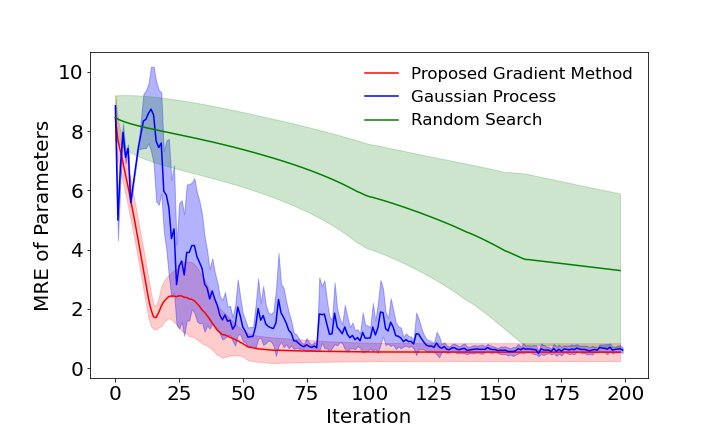}
        \caption{Parameter Estimation Error}
        \label{fig:pra_1}
    \end{minipage}
\end{figure}


{The results in Figures~\ref{fig:pre_1} and \ref{fig:pra_1} indicate that the proposed calibration approach is significantly more sample-efficient than the random design and GP-based calibration approach. The MRE for mAb production predictions using our method drops to about 60\% after 60 iterations, compared to 220\% with a random design. Furthermore, the MRE for our approach converges to about 35\% after 60 iterations for parameter estimations, significantly lower than the 700\% observed with the random policy. In contrast, the GP-based method requires around 125 iterations to reach similar convergence levels for mAb production predictions.} 

\begin{table}[h]
\centering
\caption{Comparison of runtime for various calibration methods}
\begin{tabular}{|>{\centering\arraybackslash}m{4cm}|>{\centering\arraybackslash}m{1.8cm}|>{\centering\arraybackslash}m{2cm}|>{\centering\arraybackslash}m{2cm}|>{\centering\arraybackslash}m{2cm}|}
\hline
{Time  (Hour)} & {50 Iterations} & {100 Iterations} & {200 Iterations} & {400 Iterations} \\ \hline
{Proposed Approach} & 0.131  & 0.345 & 0.791 & 1.874 \\ \hline
{Random Search}            & 0.097 & 0.186 & 0.373 & 0.86 \\ \hline
{Gaussian Process}         & $0.195$  & $0.866$ & $5.695$ & $>24$ \\ \hline
\end{tabular}
\label{tab:complexity}
\end{table}
{To assess the computational efficiency, the total runtimes are presented in Table~\ref{tab:complexity}. it indicates that the GP method’s computational time increases significantly with larger datasets due to its cubic computational complexity from matrix inversion \citep{williams2006gaussian}. In contrast, the proposed gradient method shows a sub-cubic and more manageable increase in computation time, indicating better scalability and efficiency for larger experiments.}

\section{CONCLUSION}
\label{sec: conclusion}
This study develops
a robust calibration approach for a Bio-SoS mechanistic model in the context of cell culture processes. By strategically guiding more informative data collection through the proposed gradient-based calibration approach, we significantly enhance the fidelity and predictive accuracy of the digital twin model. The empirical validation using the CHO cell culture model underscores the superiority of our approach over the traditional random design, showcasing its potential applicability across various biological systems. 
Overall, this work not only contributes to the theoretical advancements in bioprocess digital twin development but also holds promise for practical implementations that could improve the efficiency and effectiveness of biomanufacturing in the pharmaceutical industry.

\bibliographystyle{unsrtnat}
\bibliography{references} 

\end{document}